\documentclass[10pt,journal]{IEEEtran}
\IEEEoverridecommandlockouts
\usepackage{cite}
\usepackage{amsmath,amssymb,amsfonts}
\usepackage{textcomp}
\usepackage{xcolor}
\usepackage[utf8]{inputenc}
\usepackage{amsmath,amssymb}
\usepackage{float}
\usepackage{epsfig}
\usepackage{graphics}
\usepackage{graphicx}
\usepackage{color}
\usepackage{latexsym}
\usepackage{epstopdf}
\usepackage{subcaption}
\usepackage{algorithm}
\usepackage{algpseudocode}
\usepackage{pifont}
\usepackage{bbm}
\usepackage{textcomp}
\usepackage{amsthm}

\usepackage{url}

\newtheorem{definition}{Definition}

\DeclareMathOperator*{\argmax}{arg\,max}

\IEEEoverridecommandlockouts
\newtheorem{theorem}{Theorem}[]
\newtheorem{lemma}[theorem]{Lemma}

\title{Strategy-Proof Spectrum Allocation among\\ Multiple Operators  for Demand Varying\\ Wireless Networks\IEEEauthorrefmark{2}\thanks{\IEEEauthorrefmark{2}This paper is a substantially
expanded and revised version of the work in \cite{multiple_op_sp}.}}
 \author{Indu~Yadav, 
         Ankur~A.~Kulkarni,
         ~ Abhay~Karandikar
         \thanks{Indu Yadav and Abhay Karandikar are with the Department
of Electrical Engineering, Indian Institute of Technology Bombay,
Mumbai, 400076, India. e-mail: {$\lbrace$indu, karandi$\rbrace$}@ee.iitb.ac.in.
Ankur A. Kulkarni is with Systems and Control Department, Indian Institute of Technology Bombay, Mumbai,400076, India.
e-mail: kulkarni.ankur@iitb.ac.in.
Abhay Karandikar is currently Director, Indian Institute of Technology Kanpur (on leave from IIT Bombay),
Kanpur, 208016, India. e-mail:karandi@iitk.ac.in.}\vspace{-0.4cm}}

\begin{document}

\maketitle

\begin{abstract}
To address the exponentially increasing data rate demands of end users,
necessitates efficient spectrum allocation among co-existing operators in licensed and unlicensed spectrum bands to cater to the 
temporal and spatial variations of traffic in the wireless network. 
In this paper, we address the spectrum allocation problem among non-cooperative operators via auctions. The classical Vickrey-Clarke-Groves (VCG)
approach provides the framework for a strategy-proof and social welfare maximizing auction at high computational complexity, which makes it infeasible
for practical implementation. We propose sealed bid auction mechanisms for spectrum allocation which are computationally tractable and hence 
applicable for allocating spectrum by performing auctions in short durations 
as per the dynamic load variations of the network. We  establish that the proposed algorithm is strategy-proof for uniform demand.
Furthermore, for non-uniform demand we propose a algorithm that satisfies weak strategy-proofness. We also consider non-linear 
increase in the marginal valuations with demand.
Simulation results are presented to exhibit the performance comparison of the proposed algorithms with VCG and other existing mechanisms.

\end{abstract}


\section{Introduction}
With recent advancements in wireless communication technologies, the telecom market has witnessed exponential growth in data traffic in the past few decades. 
As per the current trends, mobile data traffic is expected to increase more than $5$ times by $2024$  \cite{ericsson}. Globally, Fifth Generation
(5G) technology will further escalate the amount of data traffic. With the rapid development of smart devices, 
the end user data rate requirements have also become stringent. Fulfilling the increasing number of end users with the desired Quality
of Services (QoSs) has further contributed to the crisis of limited, scarce and expensive ``spectrum''. To cater to the requirements of the end 
users, there are two possibilities: additional spectrum availability or efficient utilization of currently available spectrum.

Traditionally, the spectrum is allocated statically on lease for long durations such as one year or more to the service providers\footnote{The terminologies “service provider” and “operator” have been used
interchangeably throughout the paper}. Usually, service providers estimate the peak traffic conditions of the network and calculate the quantum of spectrum accordingly.
However, the peak traffic requirements arise sporadically in the network. This leads to underutilization of spectrum usage in the long run.
Therefore, the static allocation technique of spectrum  is 
highly inefficient in terms of spectrum utilization and not suitable to meet the requirements of the next generations networks.
Moreover, it has been shown that the traffic conditions in a wireless network vary as a function of time and location \cite{paul}. 
For instance, while on any regular day, peak traffic in residential areas  is more likely to occur in the evening,
whereas in office areas one may observe peak traffic during business hours. Thus, the wireless networks observe significant peak
to average  traffic ratio \cite{traffic_characterization}.

To address the issue of inefficient spectrum usage a computationally efficient spectrum allocation mechanism is required so that spectrum can 
be allocated dynamically considering 
the spatial and temporal traffic variations in the network. Auctions are commonly preferred for spectrum allocation among multiple operators.
In our work, we focus on computationally efficient 
spectrum allocation mechanisms for spectrum distribution among multiple operators, to ensure that the spectrum is allocated quickly as per service 
providers' demands.

In general, the spectrum is allocated among the operators using sealed bid auction format. In sealed bid auctions, interested buyers send their
valuations for the object in a closed envelope along with the demand, to the auctioneer. Thus, the privacy of the valuation and demand for the object
are ensured for each service provider. In spectrum auctions, spectrum valuation for a service provider depends on the desired bandwidth and 
 on other factors such as the number of subscribers and the services
desired by the subscribers. Hence, the spectrum valuation is a private information of an operator which is not known to the auctioneer. Generally,
the participants in any auction are selfish and are likely to misreport the actual valuation to the auctioneer if there is incentive to do so.
Hence, ensuring the strategy-proofness of auctions is of significant importance \cite{krishna2009auction}. An auction is said to be strategy-proof
if any operator does not gain on deviating from the true or actual value of their demands of the spectrum. This implies that even if 
an operator misreports its valuation, it can never achieve utility greater than that of the true valuation.

Strategy-proof auctions not only compel the participants to reveal their true valuations but also makes the process
of spectrum allocation easier for the auctioneer and the operators. The operators are neither required to perform complex computations nor they have to invest
time to determine the optimal bidding strategy to maximize their utility gains. Hence, it makes the process of resource allocation faster by removing 
the time and computational overhead. Moreover, strategy-proofness also increases the number of participants in an auction.

In spectrum auctions, three properties are of utmost importance: strategy-proofness, low computational complexity,
and optimality of allocation to maximize the social welfare \cite{roughgarden2016twenty}. Unfortunately, achieving all three properties simultaneously
in an auction is provably NP-Hard \cite{ebay}. Vickrey Clarke Groves (VCG) \cite{vickrey,clarke,groves} is a well-known mechanism which proposes a framework for
guaranteeing strategy-proof behavior in auctions with optimal allocation strategy, but it is computationally infeasible in large networks \cite{roughgarden2016twenty}. 

Various Dynamic Spectrum Allocation (DSA) mechanisms proposed in the literature are designed for single parameter environment \cite{roughgarden2016twenty},
 considering one base station per operator. In these works, it is assumed that individual base stations participate in the spectrum auction.
However, in general, each operator deploys multiple base stations (BSs) in a geographical region to cater to the services of the end users. 
The valuation of an operator depends on the number of BSs and the amount of traffic each operator has to serve. 
Unlike most of the existing works e.g., \cite{ebay, small, trust, satya}, we consider the problem of spectrum allocation at the operator level. 
Unfortunately, achieving strategy-proof behavior at 
operator level is more challenging as the BSs associated with an operator are cooperative in their behavior. Thus, an operator may misreport the valuation 
and demand at few BSs to increase the overall utility gain. 
Most of the existing works in spectrum allocation do not address this aspect of spectrum allocation problem. 

We focus on designing 
efficient strategy-proof mechanisms which are suitable for implementation in short durations to handle the spatio-temporal load variations of the network.
First, we propose a strategy-proof mechanism where the demand at each BS is of single channel.
Next, we extend the mechanism for multiple channel availability with non-uniform channel requirement (demand) across the BSs  and linearly 
increasing valuations with demand.  Here, we discuss the scenario when strategy-proofness of the mechanism may not be ensured.
Finally, we propose NUD-WSPAM where BSs of an operator may have different demands and per channel valuation is non-increasing  at
each BS. Here, we introduce the concept of weak strategy-proofness. We also prove the individual rationality, monotonicity and 
weak strategy-proofness of NUD-WSPAM.

Monte Carlo simulations are performed in MATLAB \cite{matlab} to evaluate the performances of the proposed spectrum allocation mechanisms. Using 
simulation results, social welfare and spectrum utilization of the proposed algorithms in comparison to other algorithms in the
literature e.g. \cite{small} are also evaluated. Simulations are also performed for large network sizes (i.e., large number of BSs)
to validate the applicability in practical scenarios.


\subsection{Related Work}
In this section, we review some related work on Dynamic Spectrum Access (DSA). Auction-based spectrum allocation approaches have been extensively 
studied in the
literature \cite{gandhi2008towards,ji2007cognitive,clemens2005intelligent,etkin2007spectrum,qiu2016demand,subramanian2008near,gopinathan}. 
As stated above, achieving strategy-proof optimal allocation and computational feasibility in a mechanism is NP-Hard.
In \cite{subramanian2008near}, the authors present a DSA mechanism in cellular networks which achieves near-optimal allocation for revenue 
maximization using greedy graph coloring approach. The authors in \cite{gandhi2008towards} studies real-time spectrum allocation mechanism. 
Though, the mechanisms proposed in \cite{subramanian2008near},\cite{gandhi2008towards} are computationally feasible in terms of implementation, 
they are not guaranteed to be strategy-proof. In \cite{gopinathan}, the authors propose a mechanism which ensures a certain fair chance of spectrum 
allocation along with the maximization of social welfare. In \cite{revenue}, the authors propose a revenue maximization mechanism for spectrum 
allocation. For revenue maximization, the combination of well known Vickrey-Clarke-Groves (VCG) \cite{vickrey,clarke,groves} mechanism and  
Myerson's Lemma \cite{myerson} are studied.
In \cite{ebay}, the authors proposed VERITAS, a sealed bid strategy-proof auction mechanism which follows a certain monotonicity behavior. 
The authors in \cite{small} propose another strategy-proof mechanism SMALL which groups 
non-conflicting base stations and sacrifices the base station(s) corresponding to the lowest bid in the winner group. SMALL has better allocation 
efficiency than that of the algorithm proposed in \cite{ebay}. In \cite{gandhi2008towards}, the authors propose an auction-based approach for fine grained 
(i.e., a channel is sliced into smaller frequencies) channel allocation. However, it does not satisfy the strategy-proofness property. As 
interference is one of the major concerns in wireless, the authors in \cite{clemens2005intelligent}, propose an auction based power 
allocation mechanism. However, it fails to be strategy-proof.

Both VERITAS \cite{ebay} and SMALL \cite{small} assume that the channel valuation increases linearly with the demand. In \cite{trust,promise,tames}, 
strategy-proof double auction mechanisms are studied. The authors in \cite{kasbekar,yi2015multi} 
studies auction-based approaches for DSA in cognitive networks. In \cite{ji2007cognitive}, the game-theoretic aspect of the DSA in cognitive networks is explored. 

The authors in \cite{adaptivechannel} consider adaptive-width spectrum allocation problem where the channel valuation is a non-increasing function of
the demand. 
To take the decrease in valuation with the demand into account, strategy-proof mechanism SPECIAL is proposed. Here, it is assumed that all the base
stations bid for all the channels available for auction. To improve the social welfare and revenue of VERITAS, the concept of reserve price in 
valuation is incorporated in \cite{satya}.

Most of the existing works is centered on designing a computationally feasible strategy-proof spectrum auction mechanism for non-cooperative
base station participation in auctions. Moreover,
 \cite{ebay,small,gopinathan,subramanian2008near,yi2015multi,trust,promise,tames,adaptivechannel} consider base stations with uniform channel demand.
However, only few works  \cite{ebay, small,adaptivechannel} consider multiple channel demand across the BSs. Except \cite{adaptivechannel},
all the works assume that the channel valuation scales linearly with the demand, which may not be true in general as throughput may not increase linearly 
as a function of bandwidth.

To the best of our knowledge, none of the previous works has considered the operators as the players in the spectrum auction. In comparison, in our 
work, we consider that
non-cooperative and rational operators participate in spectrum auctions and each operator has multiple BSs.
Our work also considers non-uniform channel requirement at the BSs.

\subsection{Contributions}
In this paper, we investigate the problem of spectrum allocation using sealed bid auction across multiple BSs of coexisting multiple operators in a geographical
region. We summarize our contributions as follows.
\begin{itemize}
\item We consider the problem of spectrum allocation among coexisting multiple operators in a region. The base 
stations associated with each operator are used to provide services to the end users. We formulate the problem in  multi-parameter 
environment to maximize the total social welfare of the auction. This has not been addressed in the literature so far.
\item We propose a strategy-proof spectrum allocation mechanism at operator level, where strategy-proofness holds for 
a set of valuations submitted to auctioneer corresponding to each operator.
\item We propose computationally efficient auction mechanism which are applicable to perform auction repeatedly in short durations as per 
traffic variation.
\item We propose a generalized spectrum allocation mechanism which is weakly strategy-proof even if the spectrum demands are not same
across the BSs of an operator. Further, we also consider the case where the channel valuation may not be linearly increasing with the 
demand of the channels at a base station.
\item We analytically prove that the proposed mechanism follows monotonicity, individual rationality and (weak) strategy-proofness.
\item We compare the performance of the proposed mechanism with various mechanisms in small as well as in large networks using Monte Carlo
simulations.
\end{itemize}

\par The rest of the paper is organized as follows. Section \ref{sec:sys_model} describes the system model and preliminaries of strategy-proof auctions.
In Section \ref{sec:mec_design}, we propose a mechanism for single channel allocation.  
Section \ref{sec:non-uniform_nsp} proposes an
extension of the mechanism proposed in Section \ref{sec:mec_design} and describes how it fails to be strategy-proof through an example.
In Section \ref{sec:non-uniform}, the generalized strategy-proof spectrum
allocation mechanism is presented. We summarize the proposed mechanisms in Section \ref{sec:summary}.
In Section \ref{sec:sim_results}, we evaluate the performance of proposed mechanisms through simulations. In Section \ref{sec:conclude},
we conclude the paper.

\section[\textwidth=16 cm]{System Model}
\label{sec:sys_model}

\begin{figure}[h]
\centering
\includegraphics[width = 0.27\textwidth]{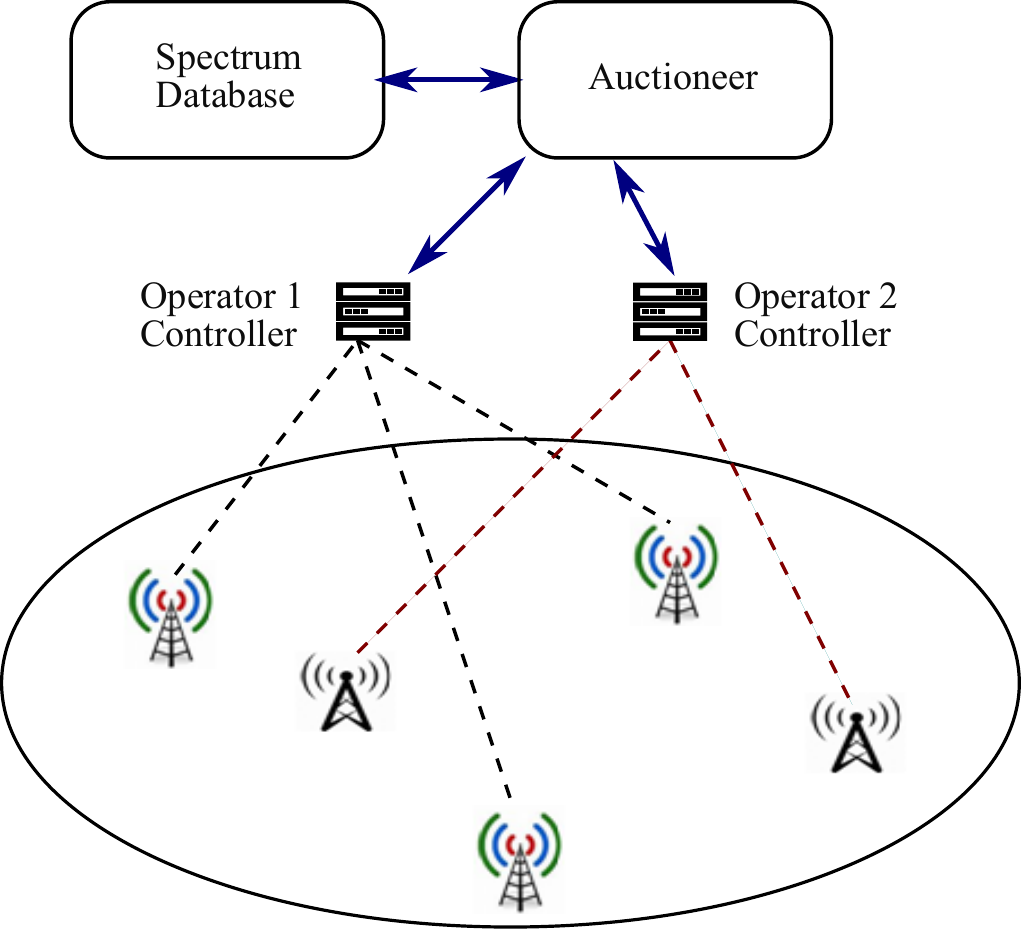}
\caption{System model.}
\label{sys_model}
\end{figure}  
We consider a geographical region where multiple operators provide services to the end users. Multiple BSs are associated with each
operator in the given region. The system model (Fig. \ref{sys_model}) comprises a 
controller for each operator, set of BSs associated with the operators, auctioneer and spectrum database. There are two decision making 
devices, controllers and auctioneer in the system. Each operator has a controller which determines the number of 
channels (demand)  required and the valuation of channels at the BSs associated with the operator.
The demand and the valuation may vary over time depending on the traffic conditions of the wireless network. The operators communicate their 
spectrum demand and valuation at each base station through the controller. The information of the number of channels
available for allocation is contained in the spectrum database. We assume that the channels are of equal bandwidth and are orthogonal. Since 
orthogonal channels do not have overlapping frequency bands, simultaneous operations on orthogonal channels do not cause interference.
Auctioneer is another decision making entity, which decides who should get the spectrum (channel) and what should be the appropriate price for 
providing exclusive {\textquoteleft right to use\textquoteright} channel to an operator.

In our work, unlike the other existing works, operators are bidders (players) instead of individual BSs in the
wireless network. Each operator communicates a vector of bids and demands to the auctioneer via the controller for the BSs associated with.

\par Other assumptions made in our system model are as follows.

 \noindent $\bullet$ We assume that an auctioneer has knowledge of the topology in the geographical region. Therefore, the overall conflict graph 
 consisting of all the BSs participating in the auction is available to the auctioneer.\\
 \noindent $\bullet$ We assume all channels are homogeneous in characteristics and act as substitutes. Thus, the bid or valuation is channel independent.\\ 
 \noindent $\bullet$ We consider that operators employ Fractional Frequency Reuse (FFR) techniques to cancel interference across its own 
 BSs. Therefore, same frequency band (channel) can be allocated to the BSs of an operator. 
 Hence, any base station of an operator would experience interference only from the BSs associated with other operators in the given region. 
\par We capture the interference among the BSs of the operators with the help of a graph $\mathcal{G} = (V,\mathcal{E})$, that is
obtained from knowledge of the topology in the geographical region, where $V$ represents the set of vertices (nodes), and $\mathcal{E}$ represents
the set of edges in the graph. The set of vertices in the graph correspond to the BSs of various operators in the region. Any two base 
stations are said to interfere with each other, if the geographical distance between them is less than a predetermined value $d$. In this case, 
there is an edge between them in the graph. Two interfering BSs (nodes) cannot be assigned same channel concurrently. 

\subsection{Background on Auctions}
\subsubsection{Strategy-Proof Spectrum Auctions}
\label{problem_form}
 
In conventional auctions, once an object is allocated to a buyer, it cannot be allocated to other buyers. However, in spectrum auctions, same 
spectrum (frequency band) can be reused or reallocated after certain fixed distance depending on the coverage area of BSs. This implies
that any two BSs can be assigned the same frequency band if they do not interfere with each other. This feature provides an advantage in 
terms of spectrum utilization, but it is more challenging to achieve strategy-proof spectrum auction. Second price auction mechanism \cite{krishna2009auction} ensures 
strategy-proof behavior in conventional auctions. However, the same is not guaranteed in the spectrum auctions \cite{ebay}. In second price 
auctions, the object goes to the highest bidder and is charged the price of the second highest bidder in the auction.
Moreover, it is not necessary that every base station of an operator interferes with each base station of other operators.
Therefore, achieving strategy-proof spectrum allocation across the multiple BSs of the coexisting operators using second price auction is
not possible. Moreover, it fails to exploit the re-usability of the spectrum which again results in inefficient usage of the spectrum.

VCG mechanism is the first strategy-proof mechanism which always chooses the optimal allocation strategy. VCG mechanism selects the set of 
participants that maximizes the overall sum of valuation in the auction \cite{vickrey,clarke,groves}. But, determining the optimal allocation and 
pricing strategy is burdened
with the high computational complexity of the auctions. Due to high computational cost, VCG mechanism is not suitable for dynamic spectrum 
allocation auctions even in wireless networks of moderate size \cite{roughgarden2016twenty}. In general, VCG mechanism is 
 applicable in combinatorial auctions for sealed bid format, where each player submits a bid for the channel without the knowledge of other 
 players bids in the auction. Unlike second price auctions, VCG is applicable for single parameter environment as well as multi-parameter 
 environment. Next, we describe the VCG mechanism for spectrum allocation.

\subsubsection{Vickrey-Clarke-Groves Mechanism}

\par We assume that there are $n$ BSs to participate in spectrum auction which leads to $2^n$ possibilities. Due to the interference across 
the BSs, all $2^n$ combinations may not be feasible for spectrum allocation. The BSs which are sufficiently far can be allocated channels simultaneously. Let the binary vector $x = \{ x_1, x_2,\ldots,x_n \}$ denote a feasible 
allocation satisfying all the interference constraints, where 
$x_i = 1 $ if a channel is assigned to the BS $i$, otherwise $x_i = 0$. Let $\mathcal{X}$ denote the set of feasible allocations. BS $i$ submits a
bid $b_i$ based on its valuation. Let $b = \{b_1,b_2,\ldots,b_n \}$. The optimal
allocation is given as
\begin{equation}
\label{eqn:eqn1.1}
 x^{\star} = \argmax_{x \in \mathcal{X}} b\cdot x.
\end{equation}

Now, a pricing scheme is defined to make the auction strategy-proof. Using a pricing scheme, the players are enforced to submit
 true valuation of the object to the auctioneer. 
VCG pricing scheme charges the BSs with the welfare loss inflicted due to the presence of BS $i$.
 
Let $\rho_i$ denote the price charged to BS $i$. 
\begin{equation}
\label{eqn:eqn1.2}
 \rho_i = \max_{x \in \mathcal{X}} \sum_{j \neq i} x_j \cdot b_j - \sum_{j \neq i} x_j^\star \cdot b_j,
\end{equation}
where $x^\star$ is the optimal allocation obtained from Equation (\ref{eqn:eqn1.1}). The price charged using Equation (\ref{eqn:eqn1.2}) also
ensures individual rationality i.e., $ 0 \leq \rho_i \leq b_i$. In other words, any BS would never be charged more than its submitted bid.
The individual rationality reflects that the utility gain at a BS can never be negative if a BS bids at its true value.

 Though VCG mechanism achieves the  optimal channel allocation for social welfare maximization, it becomes intractable for large set of 
BSs. Hence, it is not feasible for practical implementation.  
Next, we propose strategy-proof mechanisms to maximize the social welfare of the spectrum for various scenarios. The proposed algorithms are 
also computationally efficient in comparison to VCG. VCG is implemented in two steps: Channel Allocation ($\mathcal{O}(2^n)$) and Price Charging 
scheme ($\mathcal{O}(2^n)$).

\subsection{Notations and Definitions }
 We introduce the following notations:
\begin{itemize}
 \item $\mathcal{N} = \{1,2,\ldots,n\}$ represents the set of operators participating in the spectrum auction in a geographical region.
 \item $m_{i}$ represents the number of base stations corresponding to  operator $i$.
 \item $\mathcal{S}_i = \{S_{i1},S_{i2},\ldots,S_{im_i}\}$ represents the set of base stations of operator $i$.
 \item $v_{ij}$ represents the true valuation of $j^{th}$  base station corresponding to operator $i$ (i.e., $S_{ij}$).
 \item $v_i = [v_{i1},v_{i2},\ldots,v_{im_i}]$ represents the vector of true valuations at base stations of operator $i$.
 \item $b_{ij}$ represents the bid of $S_{ij}$.
 \item $b_i = \{b_{i1},b_{i2},\ldots\,b_{im_i}\}$ represents the vector of bids for operator $i$.
 \item $\mathcal{N}_i$ represents the set of neighboring base stations which are in conflict with the base stations of operator $i$ (same channel
 cannot be allocated simultaneously).
 \item $x_{i}$ represents the binary allocation vector corresponding to operator $i$.  
 \item $x_{ij}$ represents the $j^{th}$ component of $x_i$. $x_{ij} =1$ signifies channel is assigned otherwise not.
 \item $O_i$ represents operators that are neighbors of $i$ i.e.,  ($ \{ \text{operators} ~j ~|~ S_j \bigcap \mathcal{N}_i \neq \phi, j \neq i\})$.
 \item $d_i = \{d_{i1},d_{i2},\ldots,d_{im_i}\}$ represents the number of channels required at base stations of operator $i$.
 \item $N(\mathcal{G}^{'})$ represents the set of active operators from the conflict graph $\mathcal{G}^{'}$ (operators with non-zero demand).

\end{itemize}
\begin{enumerate}
 \item \textit{True valuation} ($\sigma_i^{v}$) : True valuation $\sigma_i^{v}$ of any operator $i$ is defined as the sum of the actual valuations
 (which are private and not known to the auctioneer) of all the BSs corresponding to operator $i$.
 \begin{equation}
 \label{eqn:eqn1}
 \sigma_i^{v} = \sum_{j=1}^{m_i} v_{ij}.
 \end{equation}
 \item \textit{Bidding valuation} ($\sigma_i^{b}$) : Bidding valuation $\sigma_i^{b}$ of operator $i$ is defined as the sum of the bids 
 (which may or may not be same as the actual valuation) of all the BSs corresponding to operator $i$.
  \begin{equation}
  \label{eqn:eqn2}
 \sigma_i^{b} = \sum_{j=1}^{m_i} b_{ij}.
 \end{equation}
 \item \textit{Price} ($p_{i}$): It is defined as the price that an operator $i$ has to pay, in case operator $i$ wins the resources (channels), 
 else it is zero.

 \item \textit{Operator Utility $(\mathcal{U}_i)$} : Utility of an operator $i$ is the difference between the operator valuation (unknown to the 
 auctioneer) and the price charged on the allocation of the channel. If the operator does not get the channel, the utility is zero. In other
 words, it represents the overall gain of an operator $i$ if it is allocated a channel.
 \begin{equation}
 \label{eqn:eqn4}
 \mathcal{U}_i(b_i,b_{-i}) =  \begin{cases} \sigma_i^{v} - p_{i},& ~\text{if the channel is allocated} \\ 0, & ~\text{otherwise}. \end{cases}.
 \end{equation}
\end{enumerate}
 where $b_i$ is the bid vector of operator $i$ and $b_{-i}$ represents bid vectors of other operators except operator $i$.
\begin{definition} \label{def: strategy-proofness}
 An auction is truthful (strategy-proof) if there is no incentive in  deviating  from  the  true  valuation. Thus,  the dominant  strategy  is  
 to  bid  at the true valuation no matter what strategy others choose.
\begin{equation}
 \label{eqn:eqn5}
 \mathcal{U}_{i}(b_i, b_{-i}) \leq \mathcal{U}_{i}(v_i, b_{-i}) \quad \forall b_{i}, \forall b_{-i}.
 \end{equation}
 where $v_i$ is the vector of true valuations at the base stations of operator $i$. 
 \end{definition}

\section{Strategy-proof auction for unit demand}
\label{sec:mec_design}
In this section, we describe our proposed algorithm Single Channel Strategy-proof Auction Mechanism (SC-SPAM) for channel allocation among the 
base stations of multiple operators. Recall the assumptions made 
in previous sections, multiple BSs corresponding to an operator and one channel availability. In auctions, the mechanism design has two 
steps: channel allocation and price charging strategy. In channel allocation phase the auctioneer decides to whom the right to use the channel is 
provided. What price should be charged is decided in pricing strategy phase. The price charged enforces the operators to declare the true 
valuations in order to ensure a strategy-proof auction. Now, we define critical operator which is used later in the price charging strategy.
\begin{definition} \label{def: critical neighbor}
A critical operator $C(i)$ of an operator $i$ is defined as the operator in $\mathcal{N}_i\setminus \{i\}$ whose sum of the bids of base stations is maximum among all the 
operators in $\mathcal{N}_i$ except $i$. The critical operator $C(i)$ is given as any $j \in O_i$ such that
\begin{equation}
 \label{eqn:eqn6}
 \begin{split}
  \sum_{k \in \{ \mathcal{N}_i\bigcap S_j\}} b_{jk} \geq \sum_{k \in \{ \mathcal{N}_i\bigcap S_j'\}}b_{j'k},
 \quad \forall j' \neq j,~i~  \text{and}~ j' \in O_i.
 \end{split}
 \end{equation}
 \end{definition}
 Let us define a set $\mathcal{L}_{j}^{i}= \mathcal{N}_i \cap S_j$, which contains the BSs of operator $j$ in conflict with the BSs 
  of operator $i$. Let $\Lambda_{j}^{i}$ be the valuation of set $\mathcal{L}_{j}^{i}$ which is given as,
 $\Lambda_{j}^{i} = \sum b_{jk}\mathbbm{1}_{\{S_{jk} \in \mathcal{L}_{j}^{i} \}}$. The critical operator of an operator $i$ can be obtained as
 $C(i) = \argmax \limits_{j \neq i}\Lambda_j^{i}, ~j \in O_i$ and the critical operator valuation 
 $\sigma_i^{c}$ is given as, $\sigma_i^{c} = \max \limits_{j\neq i} \Lambda_j^{i}, j \in O_i$.
 
 
 The strategy-proof algorithm proposed is described in Algorithm \ref{channel_allocation}. 
\begin{algorithm}
\caption{Single Channel Strategy-proof Auction Mechanism}
\label{channel_allocation}
\begin{algorithmic}[1]
\State \textbf{Input: }Conflict Graph $\mathcal{G}$, bid vector, $ \{ b_i \}_{\{ i \in \mathcal{N} \}}$.
 \State\textbf{Output: }Binary channel allocation vector $\{ X_i\}_{\{ i \in \mathcal{N} \}}$, price $\{p_i\}_{\{ i \in \mathcal{N}  \}}$.
 \State Initialize  $x_i \leftarrow 0$, $N(\mathcal{G})= \{ 1,2,\ldots, n\}$
\State Initialize $p_i \leftarrow 0$, $\mathcal{G^{'}} \leftarrow \mathcal{G}$, $N(\mathcal{G}^{'}) \leftarrow N(\mathcal{G})$, $FLAG \leftarrow True$.
 \While  {$(FLAG = True)$}
   \State Make $i^* \leftarrow \argmax \limits_{i \in N(\mathcal{G}^{'})}\sigma_i^{b}$.
   \State Find $\mathcal{N}_{i^*}$.
   \State Set $C(i^*) \leftarrow \argmax \limits_{j \neq i^*}\Lambda_j^{i^*},~j \in O_i*$ and $\sigma_{i^{*}}^{c} \leftarrow \max \limits_{j\neq i^*} \Lambda_j^{i^*},~j \in O_i*$. \label{pc1:cn}     
   \State Make $p_{{i}^*} \leftarrow \sigma_{i^{*}}^{c}$ and  $x_{{i}^*} \leftarrow 1$.
     \If {$ (\mathcal{G}^{'} \cap (S_{i^*} \cup \mathcal{N}_{i^*}) = \mathcal{G}^{'})$}
       \State  $ FLAG \leftarrow False. $
     \Else 
       \State $\mathcal{G}^{'} \leftarrow \mathcal{G}^{'} \backslash \{ S_{i^*} \cup \mathcal{N}_{i^*} \}$ \label{pc1:null}.
     \EndIf 
\EndWhile
\end{algorithmic}
\end{algorithm}
This algorithm takes conflict graph $\mathcal{G}$ and bid vector corresponding to each operator $ \{ b_i \}_{\{ i \in \mathcal{N} \}}$ as 
input. Binary channel
allocation vector $\{ x_i\}_{\{ i \in \mathcal{N} \}}$ and payment vector $\{p_i\}_{\{ i \in \mathcal{N} \}}$ for all the operators are initialized to zero.
Initially, we determine the maximum bidding operator and its critical neighbor $C(i^*) = \argmax \limits_{j \neq i}\Lambda_j^{i^*}, ~j \in O_i*$ (line\ref{pc1:cn}). 
Channel allocation vector, $x_i$ for the maximum bidding operator (winner) is updated to $1$ and the payment for the winning operator is 
updated to the price of the critical neighbor valuation, $\sigma_{i^*}^{c}$. The conflict graph $\mathcal{G}^{'}$ is updated with the remaining 
nodes after the removal of the nodes corresponding to the winning operator $i^*$ and its neighboring nodes $\mathcal{N}_{i^{*}}$. Repeat the
process until $\mathcal{G}^{'}$ is NULL (line \ref{pc1:null}), i.e., no other BSs is present in $\mathcal{G}^{'}$. Next, we explain Algorithm \ref{channel_allocation}
through an example.

\par \textbf{Example:} Consider a network of $3$ operators $A, B, C$, where each operator has $3$ BSs deployed in the region to provide 
services to the subscribers. BSs $\{A_1, A_2, A_3\}$, $\{B_1, B_2, B_3\}$ and $\{C_1, C_2, C_3\}$ correspond to operators $A$,
$B$ and $C$, respectively. The conflict graph is illustrated in Fig. \ref{algo_1} based on the interference criteria.
 
 \begin{figure}
        \centering
        \begin{subfigure}[b]{\linewidth}
        \centering
                \includegraphics[width=0.7\textwidth]{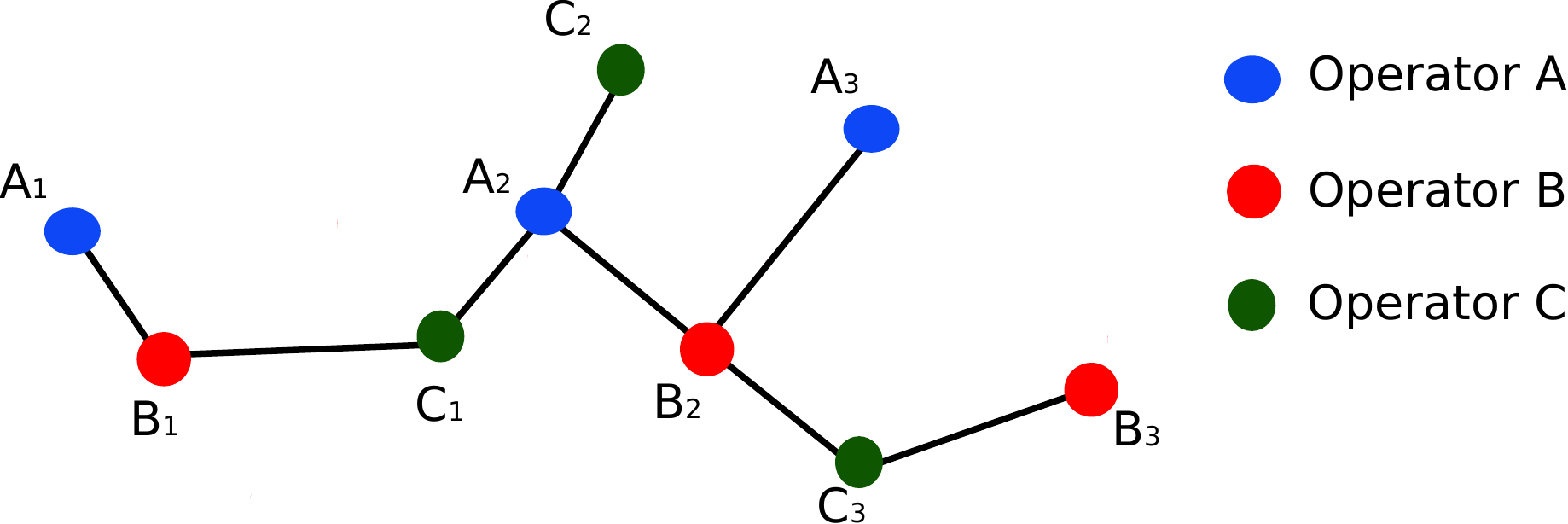}              
                \caption{}
                 \label{algo_1}
        \end{subfigure}\\
        
      \begin{subfigure}[b]{\linewidth}
      \centering
                \includegraphics[width=0.6\textwidth]{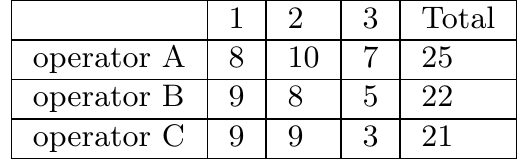}
                \caption{}
              \label{t:bid_vector}
        \end{subfigure}
        \caption{ Network of $3$ operators (a) Conflict Graph (b) Bid vector table corresponding to operator A, B and C.}
\end{figure}
 \par In Fig. \ref{t:bid_vector}, the bid vector of each operator is shown. In the first iteration, Operator $A$ has the highest bid among the operators
 with a value of 
 $\sigma_A^b = 25$. Therefore, Operator $A$ is allocated channel across BSs, and it has to pay the price of its critical operator. As per
 Definition \ref{def: critical neighbor}, critical operator for winning operator $A$ is operator $C$ and $p_A = \sigma_A^c = 18$. Thus, the
 utility of operator $A = \mathcal{U}_A = 7$. We update the conflict graph with the BSs of operators $B$ and $C$ not in conflict with the
 BSs of operator $A$.
 In second iteration, the updated $\mathcal{G}$ comprises BSs $B_3$ and $C_3$. Operator $B$ wins the channel and pays the price, $ \sigma_B^c = 3$.
 The utility of operator $B$ is $2$. Operator $C$ is not allocated channel. 
 
 Now, if operator $B$  tries to increase its utility by deviating from its true valuation $\sigma_B^v = 22$ to $\sigma_B^b = 28$ by increasing the
 bid of its BSs, operator $B$ will get channel being the highest bidder among the operators. But, it has to pay the price of its critical
 operator which is operator $A$ and therefore, pays $ \sigma_B^c = 25$. This leads to a negative utility $-3$ for operator $B$. Thus, bidding at the true 
 valuation is the best strategy for an operator in the auction. 
 Next, we prove that the proposed algorithm follows monotonicity and strategy-proofness.

 \begin{lemma}
 If  operator $i$ is allocated a channel by bidding at $\sigma_i^{b}$, it will also be allocated if it bids $\sigma_i^{b'}$, where $\sigma_i^{b'} 
 \geq \sigma_i^{b}$ provided all the other operators' bids remain unchanged. \label{lemma:monotone1}
\end{lemma}
\begin{proof}
 As stated in Algorithm \ref{channel_allocation}, all the operator bids are arranged in non-increasing order of the bids
 $\sigma_i^{b}, \forall i \in \mathcal{N}$. Let us assume in the sorted list ($S$, say) operator $i$ lies at position $k$. Now, keeping all the other operator
 bids unchanged, increase the bid of operator $i$ to $\sigma_i^{b'}$, and again arrange all the operator bids in non-increasing order in another
 sorted list $S'$. Let us say, the position of operator $i$ in $S'$ is $l$, where $l \leq k$. Thus, the operator moves higher in the position 
 which ensures that it still gets the channel. This completes the proof.
\end{proof}

\begin{lemma}
 Algorithm \ref{channel_allocation} is individually rational. \label{lemma:rational1}
\end{lemma}

\begin{proof}

As stated in the pricing scheme of Algorithm \ref{channel_allocation}, winning operator $i$ is charged price $p_i = \sigma_i^c$. Moreover, we know
winning operator $i$ valuation is the highest among all operators. 
  \begin{align}
  \label{eqn:ir1:1}
  \therefore~~~ \sigma_i^{b} > \sigma_j^{b}, ~~\forall ~j \neq i.
  \end{align}
Using Definition \ref{def: critical neighbor}, $\sigma_i^c = \max \limits_{j\neq i, ~j \in O_i} \Lambda_j^{i} $. This implies that
\begin{align}
\label{eqn:ir1:2}
 \sigma_i^c \leq \max \limits_{j \neq i} \sigma_j^{b}.
\end{align}
From Equations [\ref{eqn:ir1:1}] and [\ref{eqn:ir1:2}], we get $\sigma_i^{c} < \sigma_i^{b}$. Hence, $p_i \leq \sigma_i^{b}$. 
This  proves individual rationality of the algorithm.
\end{proof}

\vspace{-0.2cm}
\begin{theorem}Algorithm \ref{channel_allocation} is strategy-proof\label{thm:thm1}.
\end{theorem}
\begin{proof}
Refer Appendix $A$.
\end{proof}

\section{Extension for non uniform demand of channels among the base stations of operators }
\label{sec:non-uniform_nsp}
In this section, we  extend SC-SPAM for the case where the demand 
of channels  across the BSs of an operator is not uniform (or same). Instead the BSs of an operator may have different channel 
requirements depending on the traffic conditions. Let us define the demand of operator $i$ as $d_i = \{d_{i1},\ldots,d_{im_i}\}$,
where $d_{ij}$ 
represents the channel demand at $j^{\rm th}$ BS associated with operator $i$.
It is assumed that the operators do not have strict demand, i.e., they are willing to accept any number of 
channels between $0$ to $d_{ij}$ at BS $j$. 

\par Let us define, $$\sigma_i^b(l) = \sum \limits_{j=1}^{m_i}b_{ij}\mathbbm{1}_{\{d_{ij} > 0 \}}, ~l= \{1,\ldots K \}$$ for every operator $i$
where $b_{ij}$ is per channel bid value corresponding to $j^{th}$ BS of operator $i$. $\sigma_i^b(l)$ computes the valuation of each 
operator corresponding to demand of channel at its BSs for a channel. As stated above, at least one channel is required at all the BSs
participating in auction of any operator, therefore, $\sigma_i^b(1)= \sigma_i^b$ (Equation (\ref{eqn:eqn2})).

\begin{algorithm}[t]
\caption{Non-uniform Demand Auction Mechanism (NUD-AM) }
\label{variable_demand}
\begin{algorithmic}
\State \textit {\bf Input: }Conflict Graph $\mathcal{G}$, $K$ channels, bid vector, $ \{ b_{i} \}_{\{ i \in \mathcal{N} \}}$, demand vector, $ \{ d_{i} \}_{\{ i \in \mathcal{N} \}}$.
\end{algorithmic}
\begin{algorithmic}
\State \textit {\bf Output: }Allocation vector $\{ x_{i}\}_{\{ i \in \mathcal{N} \}}$, price $\{p_{i}\}_{\{ i \in \mathcal{N} \}}$.
\end{algorithmic}
\begin{algorithmic}[1]
\State Initialize demand vector $d_i' \leftarrow d_{i},$ for every $i$, $\ell= K$.
\While {$(\ell > 0)$}
 \State Compute $\sigma_i^{b}(\ell)= \sum\limits_{j=1}^{m_i}b_{ij}\mathbbm{1}_{\{d_{ij} > 0\}}$. \label{ps3:val_compute}
 \State Allocate channel and compute price (Algorithm \ref{channel_allocation}). \label{ps3:allocate}
 \State $d_i' \leftarrow d_i' - x_i$ for every $i$ \label{ps3:dmd_updt}
\Procedure{Conflict\textendash Graph\textendash Updation}{}
\State If $(d_{ij} = 0)$ update $\mathcal{G}^{'} \leftarrow \mathcal{G}^{'}\setminus \{S_{ij}\}$ 
\State Else $\mathcal{G}^{'} \leftarrow \mathcal{G}^{'}$ 
\EndProcedure
\State $\ell \leftarrow \ell-1$.
\EndWhile
\end{algorithmic}
\end{algorithm}

We propose Non-uniform Demand Auction Mechanism (NUD-AM) in Algorithm \ref{variable_demand} which takes the demand vector 
$\{ d_{i} \}_{ i \in \mathcal{N}}$ as input along with the number of channels for auction.
Here, it is assumed that the valuation of the channel increases linearly with the demand at any BS. The bid vector, 
$b_i$ reflects per channel bid for BSs of an operator. In case, the demand of the channel at any BS is $r$, then valuation at the particular 
BS gets multiplied by the demand, i.e., $r\cdot v_{ij}$. 
Channel allocation and price computation are performed iteratively for each channel present in database. For each channel allocation, we compute
$\sigma_i^{b}(l)$, which determines the operator valuation as per the demand at its BSs (line \ref{ps3:val_compute}).
Based on the operator valuation, we determine channel allocation  and price charged from the operators using SC-SPAM.
Then, demand across BSs is updated based on the allocation vector for every operator (line \ref{ps3:dmd_updt}). Next, we 
update the conflict graph before the  next channel allocation.
Channels are allocated corresponding to $\sigma_i^{b}(l)$, to ensure the maximization of the social welfare.
The process continues until all the channels are allocated. Next, we describe the operations of NUD-AM with an example. 

\subsection{Example}
\label{eg:cg2}
We consider a wireless network of $3$ operators $A$, $B$ and $C$. Each operator has multiple BSs to provide services to the end users in
a geographical region.
As illustrated in Fig.\ref{eg2}, operator $A$, $B$ and $C$ have BSs $\{A_1,A_2,A_3,A_4\}$, $\{B_1,B_2,B_3,B_4\}$ and $\{C_1, C_2\}$, 
respectively. We consider that the channel demand across the BSs of an operator is not the same, and the valuation at any BS
increases linearly with the demand.  We consider $2$ channels are available for auction. 
An operator can bid for at most the number of channels available for auction at any of its BSs. Each operator submits 
bid vector. As stated above bids are linearly increasing with demand, the bid vector contains bid per channel at each BS.

\begin{figure}[h]
 \centering
 \includegraphics[width = 0.45\linewidth]{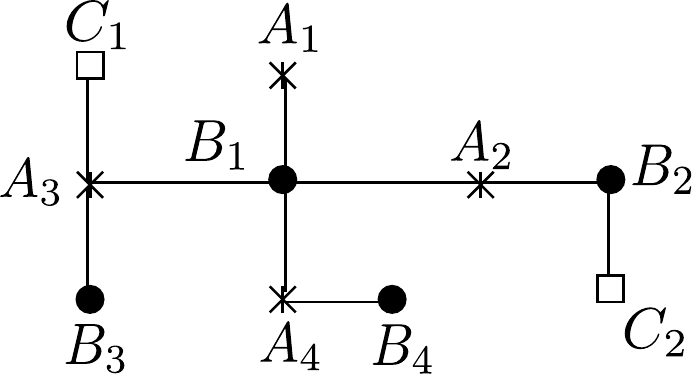}
 \caption{Conflict graph of the 3 operators. \label{eg2}}
\end{figure}

We consider the demand vectors for the operators $A$, $B$ and $C$ are given as $d_{A} = [2~ 1~ 2~ 2]$, $d_{B} = [2~1~1~2] $ and 
$d_{C} = [2~1]$, respectively. The bids at the BSs of operators $A$, $B$ and $C$ are represented as $b_A = [8~10~7~6]$,
$b_{B}=[8 ~9 ~9 ~10]$ and $b_{C}=[10 ~9]$, respectively. Channel allocation procedure is performed in two iterations.

\noindent \textbf{Case 1 : All operators bid at true value across BSs}.\\
\noindent $\bullet$ \underline{Iteration 1}: First we determine the $\sigma_{i}^{b}(1)$, $\forall~ i = \{A,B,C\}$. $\sigma_A^{b}(1)= 31$, 
$\sigma_B^{b}(1)= 36$ and $\sigma_C^{b}(1)= 19$.
Similar to the calculations shown in Section \ref{channel_allocation}, Operators $B$ and $C$ get channel at BSs
$\{B_1,B_2,B_3,B_4\}$ and $\{C_1 \}$. Now, we obtain the price charged from the winners of the auction using critical operator 
(Definition \ref{def: critical neighbor}). The price charged from operator $p_B = 31$ and $p_C = 0$. Next, we update the conflict graph 
for second channel allocation with non-zero demand across BSs as illustrated in Fig. \ref{eg1_seq1}.

 \begin{figure}[h]
 \centering
 \includegraphics[width = 0.37\linewidth]{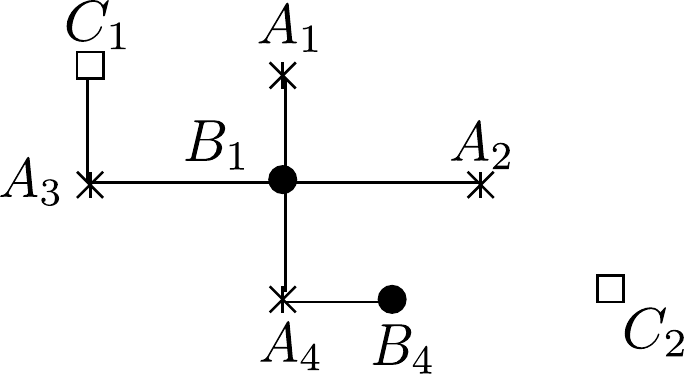}
 \caption{Updated conflict graph after the first iteration.\label{eg1_seq1}}
\end{figure}

\noindent $\bullet$ \underline{Iteration 2}: Again we perform same procedure as describe in Iteration 1 on updated $\sigma_A^{b}(2)= 31$, 
$\sigma_B^{b}(2)= 18$ and $\sigma_A^{C}(2)= 19$. Now, BSs $\{A_1, A_2, A_3, A_4\}$ and $\{C_2\}$ get channel corresponding to operators
$A$ and $C$. The price charged from the operators which are allocated channels is $p_{A}= 18$ and $p_{C}= 0$.\\

\noindent \textbf{Case 2 : Except operator $B$ all operators bid at true value}.\\
 Let Operator $B$ deviates from the true valuation and submits  $b_B= (8,~6,~6,~9)$ to the auctioneer. \\
\noindent $\bullet$ \underline{Iteration 1}: As Operator $B$ deviates from the true value, $\sigma_{B}^{b}(1)$ reduces to $29$. Channels are 
allocated at $\{A_1,A_2,A_3,A_4\}$ and $\{C_2\}$ BSs of operators $A$ and $C$, respectively. The price charged are $p_{A} = 29$ and 
$p_C = 0$. Next, update the conflict graph. \\
 \begin{figure}[h!]
 \centering
 \includegraphics[width = 0.37\linewidth]{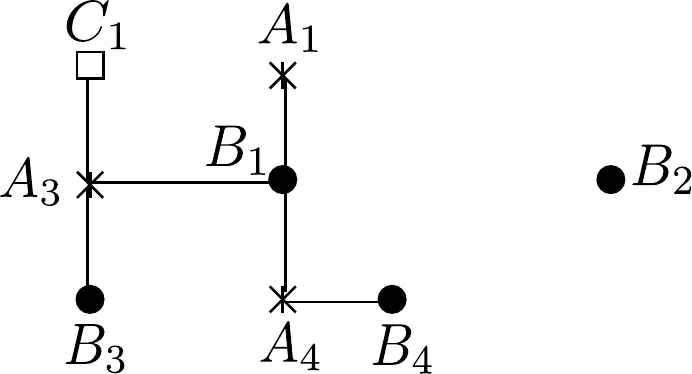}
 \caption{Updated conflict graph after the first iteration. \label{eg1_seq2}}
\end{figure}
\noindent $\bullet$ \underline{Iteration 2}: Channels are allocated on the updated graph shown in Fig. \ref{eg1_seq2} at $\{B_1,B_2,B_3,B_4\}$
and  $\{C_1\}$ BSs of operators $B$ and $C$, respectively. We observe that the demand at the BS $A_2$ is zero, so it is no 
longer the part of the conflict graph. Therefore, the price charged from the operator $B$ and $C$ are $21$ and $0$, respectively.

It is observed that operator $B$ gets the same number of channels in both the cases (true valuation and misreporting to lower value).
However, the price charged at the true value and the deviated bid value are $31$ and $21$, respectively for operator $B$. This clearly shows the 
utility gain of operator B is $10$. Hence, NUD-AM is not always strategy-proof.
 
As channel allocation procedure for a channel in NUD-AM is same as SC-SPAM, therefore, NUD-AM is strategy-proof individually for 
every iteration. But it may not be strategy-proof as a whole.
The reason behind NUD-AM not being strategy-proof is the updation of the conflict graph after each allocation. This results 
in removal of BSs where demand is satisfied. This shows that addressing non-uniform demand is challenging. Next, we present a new algorithm for 
this purpose.


\section{Weakly Strategy-proof Algorithm for Non-uniform Demand}
\label{sec:non-uniform}
Algorithm NUD-AM proposed in Section \ref{sec:non-uniform_nsp} considers non-uniform demand across the base stations of an operator where
the channel valuations increase linearly with the demand at the base stations. As NUD-AM charges price sequentially from the BSs in each step,
it fails to be 
strategy-proof in certain cases, e.g. if an operator chooses to bid lower than its true valuation.
In this section, we propose Non-uniform Demand Weakly Strategy-proof Auction Mechanism (NUD-WSPAM) which ensures that the operators have no 
incentive to deviate from the true valuation even if the demand
across BSs is non-uniform and the bids at a BS is not linearly increasing function of the demand. 
 
Moreover, in comparison to the NUD-AM, we consider that the channel valuations  are not a linear function of the demand. 
Hence, the operators are required to report the bid valuation 
corresponding to multiple channel demand at each BS to the auctioneer. We consider that the demand at any BS across the network 
cannot be greater than the total number of channels available in the spectrum database. 
Now, each operator reports a bid vector for each BS associated with it. 
Let $\mathcal{B}_i$ denote the bid for operator $i$. Here, $\mathcal{B}_i(\ell,j)$ is bid for demand $\ell$ at BS $j$ of operator $i$ 
if $(\ell-1)$ channels are already assigned.
We enforce that the bids submitted by operators are non-increasing i.e.,
$$\mathcal{B}_i(\ell,j) \geq \mathcal{B}_i(\ell+1,j),~ \text{for all}~ i,j,\ell.$$
This is motivated by the observation that marginal true value per channel is also non-increasing with the demand \cite{chandra2008case}.

\begin{algorithm}[t]
\caption{Non-uniform Demand Weakly Strategy-proof Auction Mechanism (NUD-WSPAM)}
\label{algo:sp_variable_bid}
\begin{algorithmic}
\State \textit{\bf Input: }Conflict Graph $\mathcal{G}$, $K$ channels, non-increasing bid vector, ${\mathcal{B}_{i}}_{\{ i \in N \}}$, demand vector $ \{ d_{i} \}_{\{ i \in N \}}$.
\end{algorithmic}
\begin{algorithmic}
\State \textit{\bf Output: }Channel allocation vector $\{ x_{i}\}_{\{ i \in \mathcal{N} \}}$, price $\{p_{i}\}_{\{ i \in \mathcal{N} \}}$
\end{algorithmic}
\begin{algorithmic}[1]
\State Initialize final allocation vector $x_i^{f} \leftarrow 0$, $\mathcal{G^{'}} \leftarrow \mathcal{G}$
\State Initialize $p_{i} \leftarrow 0$, $b_i= \mathcal{B}_i(1,:) \quad \forall i \in N$
\While {$(K > 0)$}
  \State Find $x_{1},\ldots,x_{n}$ using Algorithm \ref{channel_allocation} \label{ps4:allocation}
  \State Update $x_i^{f}  \leftarrow x_i^{f} +  x_i$, $\forall ~i$
  \State Update $d_i \leftarrow d_i  -  x_i^{f}$, $\forall ~i$ 
  \Procedure{Bid \textendash Updation}{}
  \State for $i = 1: n$
  \State for $j = 1: m_i$ 
  \State $b_{ij} = \mathcal{B}(x_{ij}^{f}+1, j)$ end end
 \EndProcedure
 \Procedure{Conflict\textendash Graph\textendash Updation}{}
 \State see Algorithm \ref{variable_demand}
  \EndProcedure
  \State $K \leftarrow K-1$
\EndWhile
  \State Charge price as per the Equation (\ref{eqn:cr_p2}). \label{ps4:price}
\end{algorithmic}
\end{algorithm}

In Algorithm \ref{algo:sp_variable_bid}, we present a generalized algorithm NUD-WSPAM. Unlike previous mechanisms, NUD-WSPAM first determines 
allocation for all the channels present in the spectrum database and then computes the price to be charged. At each iteration, a channel is allocated using Algorithm
\ref{channel_allocation} for every channel available in spectrum database as mentioned in line \ref{ps4:allocation}.
To compute the price, we update the conflict graph which comprises of the BSs where channel requirement is not satisfied after the 
allocation is complete. The price is charged  based on the critical operator in the final updated graph (line \ref{ps4:price}).

As described in Algorithm \ref{algo:sp_variable_bid}, allocation is performed iteratively for each channel and then the bids are
updated after each allocation for all the operators. The bid of operator $i$ is  $\mathcal{B}_i$, 
where BS $j$ has multiple bids given as $\{\mathcal{B}_i(\ell,j)| 0 < \ell \leq d_{ij}\}$.
Let $b_i^r$ denote the active bids (maximum of bid for demand that is not satisfied) at BSs of operator $i$ in $r^{\rm th}$ iteration. 
The bid updation process is described in the Algorithm \ref{algo:sp_variable_bid}. We denote $b_i^{f} = b_i^L$, where $L$ is the
last iteration. Therefore, bid vector $b_i^f$ projects the bids at the BSs of operator $i$ for $(K+1)^{\rm th}$ iteration, where $K$ is the number 
of channels available.
The bid at BS $j$ of operator $i$ in vector $b_i^f$ is given as $b_{ij}^f = \mathcal{B}_i(x_{ij}^f +1,j)$, where $x_{ij}^f$ is final allocation of 
operator $i$ at BS $j$ or $j^{\rm th}$ component of $x_i^{f}$.
The vector $b_i^f$ has the highest bid
values corresponding to unsatisfied demand (non-increasing bid assumption) for operator $i$.

Let $d_i^f$, $i \in O$ denote final demand vector of the operator $i$ after the allocation process is complete. 
Here, $d_{ij}^f = 0$ signifies that the demand is satisfied  at $j^{th}$ BS of operator $i$. 
Furthermore, the set of BSs where demand is unsatisfied is indicated $S_i^f$ i.e., $S_i^f  = \{j | d_{ij}^f > 0\}$.
Based on $S_i^f $, final conflict graph $\mathcal{G}^f$ is obtained. $\mathcal{G}^f$ has BSs where demand is not satisfied. 

Let, $\Gamma_{j}^{i} = \mathcal{N}_i\bigcap S_j^f$ denote the BSs of operator $j$ in $\mathcal{G}^f$ which are in neighborhood of
BSs of operator $i$ in initial conflict graph $\mathcal{G}$. We define the critical operator  $C(i)$  any $j \in O_i$  such that 
\begin{equation}
 \label{eqn:cr_op2}
 \begin{split}
  \sum_{k \in \Gamma_j^i } b_{jk}^f \geq \sum_{k \in  \Gamma_{j'}^i}b_{j'k}^f,
 \quad \forall j' \neq j,~ j' \in O_i.
 \end{split}
 \end{equation}
 
For single channel auction, Equation (\ref{eqn:cr_op2}) reduces to Definition \ref{def: critical neighbor}. The only difference is that the
BSs where the demand is zero after allocation process are no longer part of the conflict graph $\mathcal{G}^f$. 
We compute valuation of operator $j$ which are not allocated channel $\chi_{j}^i $.
Critical operator valuation $\sigma_i^c$ is obtained using Equation (\ref{eqn:cr_p2}).
\begin{align}
\label{eqn:eq11}
\chi_{j}^i &= \sum_{k \in \Gamma_j^i } b_{jk}^f. \\
\label{eqn:cr_p2}
 \sigma_i^c &= \chi_{C(i)}^i.
\end{align}
The price charged from operator $i$ is $p_i = \sigma_i^c$. 
This price  reduces to the earlier critical operator valuation mentioned in Section \ref{sec:mec_design} for single channel scenario.

Here we define a new concept of {\it{weak strategy-proofness}}:
\begin{definition}\label{def:weak-sp}
Let $\mathcal{B}_i$ denote true valuation of operator $i$. An auction is said to be weakly strategy-proof if an operator does not gain by deviating
to $\mathcal{\tilde{B}}_i$  from  $\mathcal{B}_i$, where $\mathcal{\tilde{B}}_i$ satisfies either $(1)$ $\exists$ $j$ such that
$\mathcal{\tilde{B}}_i(\ell,j)  > \mathcal{B}_i(\ell,j)$, $\forall \ell$ 
or (2) $\exists$ $j$ such that $\mathcal{\tilde{B}}_i(\ell,j)  < \mathcal{B}_i(\ell,j)$,  $\forall \ell$. i.e.,
\begin{equation}
\label{eqn:weak-sp}
 \mathcal{U}_i(\mathcal{\tilde{B}}_i, \mathcal{B}_{-i}) \leq \mathcal{U}_i(\mathcal{B}_i,\mathcal{B}_{-i}) \quad \forall~\mathcal{\tilde{B}}_i \& \mathcal{B}_{-i}.
\end{equation}
where, $\mathcal{\tilde{B}}_i$ satisfy conditions $(1)$ or $(2)$  and 
$\mathcal{B}_{-i}= \{\mathcal{B}_{1},\ldots,\mathcal{B}_{i-1},\mathcal{B}_{i+1},\ldots,\mathcal{B}_{n}\}$ is tuple with bid of all other 
operators except operator $i$.
\end{definition}
\subsection{Example}
\label{eg:cg3}
We revisit the Example \ref{eg:cg2} in context of NUD-WSPAM. The wireless network is illustrated in Fig. \ref{eg2} is same except the  channel valuation at a 
BS is no longer linearly increasing with the demand. As stated earlier, per channel valuation is non-increasing
function of demand at any BS. An operator can bid for at most the number of channels available for auction at any BS. We 
consider demand vectors to be same as mentioned in the example previously. 
Let $q_{ij}$ represents the bid vector at BS $j$ of operator $i$ corresponding to its demand. 
The bid at BSs of operator $A$ are given as $\mathcal{B}_{A}= [q_{A1}^{\rm T} ~ q_{A2}^{\rm T} ~ q_{A3}^{\rm T} ~q_{A4}^{\rm T}]$,  
where $q_{A1} = [8~ ~5]$, $q_{A2} = [10~ ~0]$ and $q_{A3} = [7~~3]$ and $q_{A4} = [6 ~~ 3]$. Here, $a^{\rm T}$ indicate the transpose of $a$. 
The bid for operator $B$ is $\mathcal{B}_{B} = [q_{B1}^{\rm T}~ q_{\small{B2}}^{\rm T} ~ q_{B3}^{\rm T} ~q_{B4}^{\rm T}]$, where $q_{B1} = [8 ~4]$,
$q_{B2} = [9 ~0]$, $q_{B3} = [9 ~0]$ and $q_{B4} = [10 ~ 3]$. The bid for operator C is $\mathcal{B}_{C} = [q_{C1}^{\rm T}~ q_{C2}^{\rm T}]$, 
where $q_{C1} = [10 ~5]$, $q_{C2} = [9 ~0]$. 

\noindent \textbf{Case 1: All operators reveal their true valuations} \\
\noindent $\bullet$ \underline{Iteration 1}: From the given bid vectors, we determine the bids of operators for the allocation $\sigma_{A} = 31$,
$\sigma_{B} = 36$ and $\sigma_{C} = 19$. The channel is allocated at all the BSs of the highest bidding operator. Then channel is 
allocated to the BSs of the remaining operators in the order of decreasing valuations which do not conflict with the BSs 
that are already allocated channel. Therefore, the channel is allocated to operator $B$ at $\{B_1,B_2,B_3,B_4\}$ and operator $C$ at $\{C_1\}$ 
BSs.\\
\noindent $\bullet$ \underline{Iteration 2}: For second channel allocation, demand and bid vectors are updated depending on the allocation in 
previous iteration. The updated demand vectors are $d_{A} = [2~ 1~ 2~ 2]$, $d_{B} = [1 ~0 ~0 ~1] $ and $d_{C} = [1 ~1]$. The operators 
valuation for the iteration is determined from the updated bid $\sigma_{A} = 31$, $\sigma_{B} = 7$ and $\sigma_{C} = 19$. Channel is allocated to
operator $A$ at $\{A_1,A_2,A_3,A_4\}$ and operator $C$ at $\{C_2\}$ BSs.\\

This completes the channel allocation phase. Now, the demand at the operators is $d_{A} = [1~ 0~ 1~ 1]$, $d_{B} = [1~0~0 ~1] $ and $d_{C} = [1~0]$.

\begin{figure}[h]
 \centering
 \includegraphics[width = 0.30\linewidth]{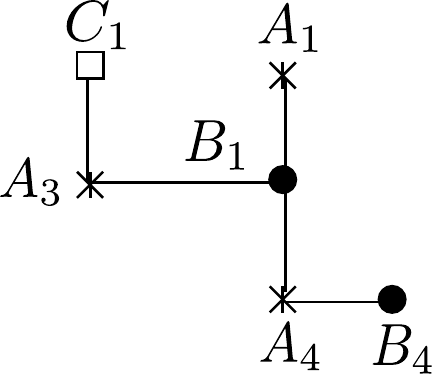}
 \caption{Updated conflict graph after channel allocation phase is complete.\label{eg1_last}}
\end{figure}

Price Charging Step: In Algorithm \ref{algo:sp_variable_bid}, the price is charged after all the channels are allocated based on the BSs where
demand is non-zero. We construct the conflict graph with the BSs having demand greater than zero as illustrated in Fig. \ref{eg1_last}. 
Each operator is charged as per their critical operator (see Definition \ref{def: critical neighbor}). The sum of the highest bids of the 
BSs $\{B_1, B_4\}$ of operator $B$ for which demand is not satisfied comprise the critical operator of operator $A$. Similarly, the bids
of BSs $\{A_1,A_3,A_4\}$ constitute the critical operator for operator $B$ and the bid at the BS $\{A_3\}$ is critical operator
bid for operator $C$. Thus, the price charged from operator $A$, $B$ and $C$ is given by $p_A = 7$, $p_B = 11$ and $p_C = 3$.

\noindent \textbf{Case 2: Operator $B$ deviates from true valuation and bids at a lower value} \\
Now, we revisit the wireless network mentioned in Fig. \ref{eg2}, considering that except the operator $B$ others submit bid equal to true value
for the associated BSs. The demand vector of all the operators remain unchanged as in the first case. We consider that operator bid is 
$\mathcal{B}_{B}^{\rm '} = [{q_{B1}^{\rm '}}^{\rm T},~ {q_{B2}^{\rm '}}^{\rm T},~ {q_{B3}^{\rm '}}^{\rm T}, {q_{B4}^{\rm '}}^{\rm T}]$, where $q_{B1}^{\rm '}
= (8, ~4)$, $q_{B2}^{\rm '} = (6,~ 0)$, $q_{B3}^{\rm '} = (6,~ 0)$ and $q_{B4}^{\rm '} = (9,~ 3)$. As described in Case $1$, channel is allocated
to the operators.

\noindent $\bullet$ \underline{Iteration 1}: Here, the operator bids for channel allocation are $\sigma_{A} = 31$, $\sigma_{B} = 29$ and
$\sigma_{C} = 19$. Operators $A$ and $C$ are allocated channel at the BSs $\{A_1,A_2,A_3,A_4\}$ and $\{C_2\}$.\\
\noindent $\bullet$ \underline{Iteration 2}: Second channel is allocated to BSs $\{B_1,B_2,B_3,B_4\}$ and $\{C_1\}$.

We can see that operator $B$ gets channel at their BSs in iteration $2$. Channel allocation remains same even after deviating from true
valuation. Next, we determine the price charged from the operators.
\par Price Charging Step: We update the conflict graph based on the remaining channel demand across the BSs of every operator as
illustrated in the Fig. \ref{eg1_last}. Then, we determine the price charged from every operator based on the critical operator. The price charged
remains the same as it is obtained for Case $1$ (operators reveal their true valuations).

From the above example, it is seen that the deviation from true valuation does not provide utility gain. Thus, operators have no incentive
in misreporting the true valuation. Hence, Algorithm \ref{algo:sp_variable_bid} is strategy-proof.

As we defined earlier, $\sigma_i^b(k) = \sum\limits_{j} b_{ij}^k$, where $b_{i}^{k}= [b_{i1}^k \ldots b_{im_i}^k]$ has the bids at which operator
$i$ demands channel at its BSs in $k^{\rm{th}}$ iteration of allocation.

\begin{lemma}
 Algorithm \ref{algo:sp_variable_bid} is individually rational.\label{lemma:rational2}
\end{lemma}
\begin{proof}

As per the assumption, marginal bid per channel decreases with the demand $\ell$ at any BS i.e., 
$\mathcal{B}_i(\ell,j) \geq \mathcal{B}_i(\ell',j)$ for $\ell < \ell'$ for all operator $i$, BS $j$. Therefore, each operator bid is non-increasing 
sequentially in the allocation process i.e., $\sigma_i^{b}(k) \geq \sigma_i^{b}(k')$ for $k < k'$, where $k$ denotes channel allocation iteration.

Operator $i$ with maximum bid gets channel in each iteration. As stated in Algorithm \ref{algo:sp_variable_bid}, updated graph
$\mathcal{G}^f$ comprises BSs $\{s | s \in S_j^f, \forall j \}.$ Operator $i$ is charged as $\sigma_i^c = \max\limits_{j \neq i}\chi_j^i$. 
As proved in Lemma (\ref{lemma:rational1}), $\sigma_i^c(k) \leq \sigma_i^b(k)$, for all $k$. But, in Algorithm \ref{algo:sp_variable_bid}, 
$\sigma_i^c$ is determined from $\mathcal{G}^f$. Therefore, $\sigma_i^c \leq \sigma_i^c(k), \forall k$.

Let $x_i^{f}$ denote the final allocation vector for operator $i$. We denote the sum of the channel bids corresponding to allocation vector
$x_i^{f}$ is $\alpha_i^{b}$. Thus, $\alpha_i^{b}= \sum\limits_{j= 1}^{m_i} \sum\limits_{\ell= 1}^{x_{i}^{f}(j)} \mathcal{B}_i(\ell,j)$. Moreover,
$\alpha_i^{b} \geq \sigma_i^b(1)$, where $\sigma_i^b(1)$ is operator $i$ bid for first channel. As we know $\sigma_i^c < \sigma_i^b$, 
therefore $\sigma_i^c < \alpha_i^{b}$.
Now, the price charged is given by
 \begin{align*}
  p_{i} &= \alpha_i^{b}- \sigma_{i}^{c},\\
        &\leq \alpha_i^{b}.        \quad \quad \quad (\because 0 \leq \sigma_i^c \leq \alpha_i^{b}).
 \end{align*}
Thus, $ 0 \leq p_i \leq \alpha_i^{b}$. This proves that the Algorithm \ref{algo:sp_variable_bid} is individually rational.
\end{proof}
\begin{lemma}
 In Algorithm \ref{algo:sp_variable_bid}, suppose final allocation vectors of operator $i$ are $x_i^f$ and $\tilde{x}_i^f$ at bids 
 $(\mathcal{B}_i,\mathcal{B}_{-i})$ and $(\mathcal{\tilde{B}}_i,\mathcal{B}_{-i})$, respectively. If there exists some $j$
 $\mathcal{\tilde{B}}_i(\ell,j) > \mathcal{B}_i(\ell,j)$ ~ $\forall \ell$, then 
 $\tilde{x}_i^f -x_i^f \geq 0$. This implies that the number of channels allocated across the BSs of operator $i$ at $\mathcal{\tilde{B}}_i$ are 
 atleast equal to the number of channels allocated at $\mathcal{B}_i$.  
\label{lemma:monotone2}
\end{lemma}
\begin{proof}
As per the assumption in Section \ref{sec:non-uniform},
$\mathcal{\tilde B}_{i}(\ell,j) \geq \mathcal{\tilde B}_{i}(\ell',j) $ such that $\ell < \ell'$ for all BSs $j$. 
Let operator $i$ be allocated channels in $k$ iterations in the allocation process at $\mathcal{B}_{i}$. As channel allocation is 
performed greedily based on the bid, with bid $\mathcal{\tilde B}_{i} \geq  \mathcal{B}_{i}$, it must be allocated at least $k$ iterations. 
Since $\mathcal{B}_{-i}$ is unchanged, operator $i$ may get a channel in  more than $k$ iterations, if increase in bid results in
$\sigma_i^b > \sigma_j^b$, for $j \neq i$ in more iterations in the allocation process. 
\end{proof}
\begin{theorem}
Algorithm \ref{algo:sp_variable_bid} is weakly strategy-proof \label{thm:thm3}.
\end{theorem}
\begin{proof}
Refer Appendix $B$.
\end{proof}
 
\section{Summary of the Proposed Mechanisms}
\label{sec:summary}
In this section, we summarize the key features (strategy-proofness and computational 
complexity) of the proposed mechanisms and the various scenarios in Table \ref{tab:summary}.

\begin{table}[h]
\caption{Summary}
\centering
\begin{tabular}{| p{1.3cm} | p {3.5cm}| p{1cm}| p{1.6cm}|}
 \hline
  \textbf{Algorithms} & \textbf{Scenario} &  \textbf{Strategy-proof} & \textbf{Computational Complexity} \\
 \hline
  SC-SPAM & Single channel, Uniform demand  &  Strong & $\mathcal{O}(mn^2)$ \\ 
 \hline
  NU-AM & Multi-channel, Non-uniform demand,linearly  bid & No& $\mathcal{O}(K\cdot mn^2)$\\
 \hline
 NUD-WSPAM & Multi-channel, Non-uniform demand, non-linear bid  & Weak&  $\mathcal{O}(K\cdot mn^2)$\\ 
 \hline
  
\end{tabular}
\label{tab:summary}
\end{table}
 In above Table, $n$ is the number of operators, $m = \sum\limits_{i=1}^{n} m_i$ is the total number of BSs across all the operators 
 present in the region, and $K$ is the number of channels available for auction. The detailed computation complexity analysis of SC-SPAM is 
 presented in \cite{multiple_op_sp}. For $K$ channel availability, computational complexity becomes $\mathcal{O}(K\cdot mn^2)$ using similar
 analysis as given for SC-SPAM.

\section{Simulation Results}
\label{sec:sim_results}
In this section, we evaluate the performance of the proposed algorithms in multi-operator settings in a wireless network. 
In the simulations, we consider $3$ operators providing services in a  
region. We model the wireless network by creating conflict graphs $\mathcal{G} = (V,\mathcal{E})$ using the configuration model~\cite{chung2002connected}.
To create overall topology of the wireless network in a given region, we first generate three conflict graphs $\mathcal{G}_{12}$,
$\mathcal{G}_{13}$, $\mathcal{G}_{23}$. Here, $\mathcal{G}_{ij}$  represents conflict graph among the BSs of operators $i$ and $j$.
Using the conflict graphs, we obtain corresponding binary interference matrices 
$\mathcal{I}_{12}$, $\mathcal{I}_{13}$ and $\mathcal{I}_{23}$, where $\mathcal{I}_{ij}$ represents the interference among the BSs of 
operator $i$ and operator $j$. 
In an interference matrix, $1$ indicates interfering pair of BSs.
Further, we obtain interference matrices $\mathcal{I}_{ji}$ from the transpose of matrix $\mathcal{I}_{ij}$.
The overall interference matrix $\mathcal{I}$ of wireless access network in the region is obtained
using  $\mathcal{I}_{12}$, $\mathcal{I}_{13}$, $\mathcal{I}_{23}$, $\mathcal{I}_{21}$, $\mathcal{I}_{31}$, and $\mathcal{I}_{32}$.
We perform Monte Carlo simulations for various scenarios. All the results are obtained by averaging over $50$ different topologies. All simulations are 
performed in MATLAB \cite{matlab}.

 We evaluate the performance of the algorithms based on the following parameters:
 \begin{itemize}
\item Spectrum utilization: It is defined as the total number of BSs which are assigned channels across all the operators, i.e., 
$\sum\limits_{i=1}^{n}\sum\limits_{j=1}^{m_i} x_{ij}$, where $x_{ij}$ denote the allocation at $j^{\rm th}$ BS of operator $i$.
\item Social welfare: Social welfare is defined as the sum of the bids corresponding to the BSs which are allocated channels. 
\end{itemize}

We compare the proposed algorithms with VCG \cite{roughgarden2016twenty} and SMALL \cite{small} mechanisms.
As discussed earlier, VCG mechanism chooses an allocation with the highest social welfare (optimal) from the set of 
all the feasible allocations. 
SMALL  groups the non-conflicting BSs together and determines the group valuation for each group. 
The group valuation is obtained as the number of BSs with bid greater than the minimum bid of the group times the minimum 
bid. Channel is allocated to the highest bidding group and all the BSs except the one with minimum bid are charged with the minimum bid in
the group. 

\subsection{Performance evaluation of SC-SPAM}
The bids across the BSs are uniformly distributed in the interval $[15, 25]$ for each operator. As VCG becomes computationally intractable for 
large networks, we restrict our simulations to small size networks which vary from $6$ to $21$ BSs. In this case, single channel is available in 
spectrum database. In Fig. \ref{perform_comp}, we observe that the social welfare and the spectrum 
utilization of SC-SPAM are close to the optimal obtained from VCG. However, SC-SPAM outperforms SMALL in both spectrum utilization and 
social welfare. 
\begin{figure}[h]
    \centering
    \begin{subfigure}[b]{0.425\linewidth}
        \centering
       \includegraphics[width=\textwidth]{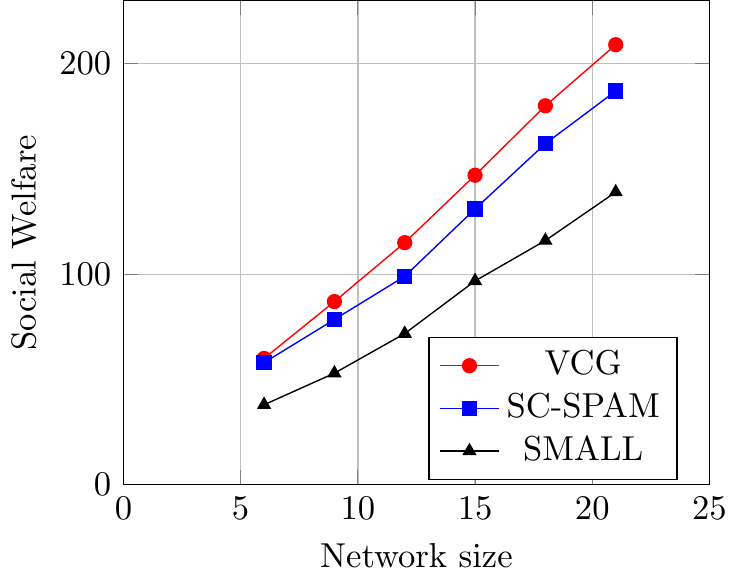}
        \caption{{Social Welfare}}
         \label{fig:sw_small_nw}
        \end{subfigure}%
    ~
    \begin{subfigure}[b]{0.425\linewidth}
        \centering
       \includegraphics[width=\textwidth]{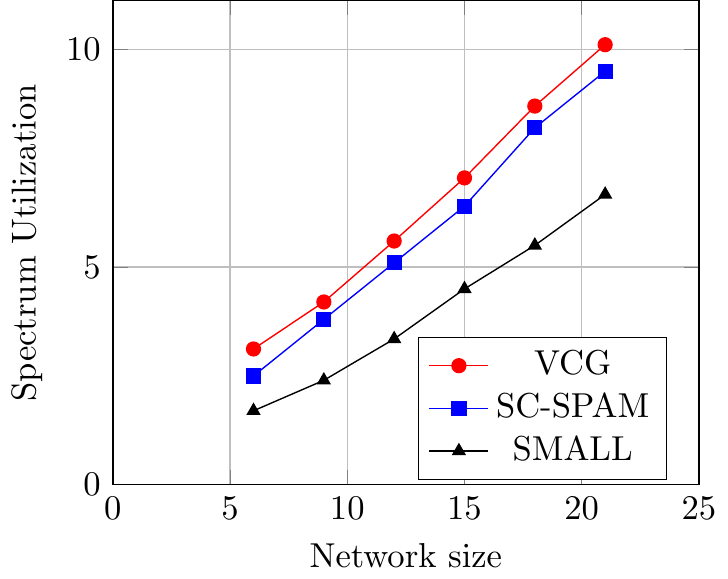}
        \caption{{Spectrum utilization}}
        \label{fig:su_small_nw}
        \end{subfigure}
        \caption{Performance comparison of the VCG, SC-SPAM and SMALL in three operator scenario.}
\label{perform_comp}
 \end{figure}
\vspace{-0.3cm}
\subsection{Performance evaluation of SC-SPAM for multiple channels}
Next, we compare the performance of SC-SPAM with SMALL \cite{small} in large networks with the number of BSs ranging from $30$ to $300$. We 
consider that
$3$ channels are available in the spectrum database. Each BS has a demand of $2$ channels for all the operators. Each operator submits a 
bid vector which has per channel valuation at every base station. The operators choose bids uniformly between $[10,25]$. From Fig. 
\ref{fig:uniform}, we observe that the performance of the proposed mechanism for multiple channel allocation is better than that of SMALL. Here, 
spectrum utilization is determined as the total number of channels allocated across the BSs of all the operators. The trend observed 
justifies the following facts: First, SMALL sacrifices the BSs with minimum bid to achieve strategy-proofness, resulting in lower social welfare. 
Second, BSs only in winning groups are allocated channel, even though there may be some BSs which do not conflict with the winning BSs. Further,
it is seen that the performance of SMALL degrades  with an increase in the number of  BSs in the region.
\begin{figure}[h]
    \centering
    \begin{subfigure}[b]{0.425\linewidth}
        \centering
       \includegraphics[width=\textwidth]{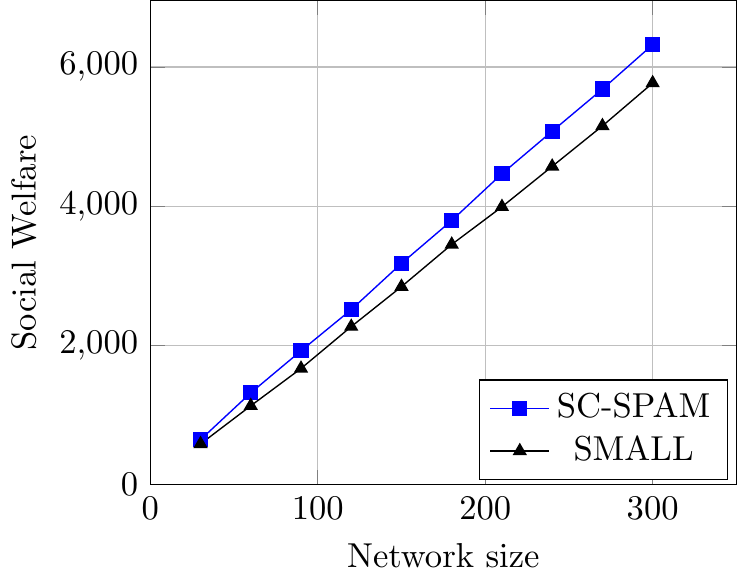}
        \caption{{Social Welfare}}
        \label{fig:sw2_uniform}
        \end{subfigure}%
    ~~
    \begin{subfigure}[b]{0.425\linewidth}
        \centering
       \includegraphics[width=\textwidth]{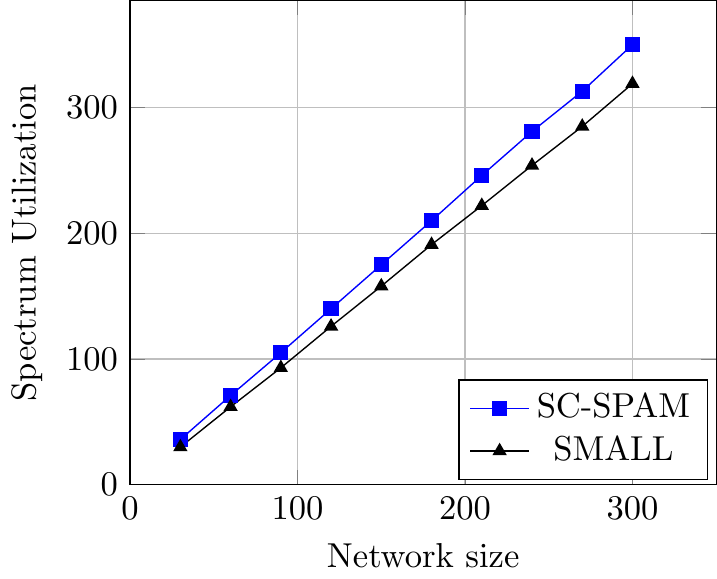}
        \caption{{Spectrum utilization}}
        \label{fig:su2_uniform}
        \end{subfigure}
        \caption{Performance comparison for uniform demand $d =2$ and $3$ number of available channels for auction across the BSs of multiple 
        operators.}
\label{fig:uniform}
 \end{figure}
 \vspace{-0.3cm}
\subsection{Performance evaluation of NUD-WSPAM} 
 We consider channel demand at any BS to be function of the traffic in the cell. 
 The demand at any BS is uniformly distributed in the interval $[0,3]$.
 We perform simulations to evaluate the number of channels required to satisfy the demands across the BSs for all operators in the region. 
 In Fig.\ref{fig:c_vs_su_nu}, we observe that the number of channels required to fulfill the demand for all the operators shows similar trend 
 irrespective of the number of BSs. The number of channels required for the wireless network of $150$ BSs remains same as that of
 $300$ BSs. 
 The reason for this behavior is that the degree distribution of BSs does not change with the size of the network (number of BSs).
 \begin{figure}[h]
 \centering
 \includegraphics[width = 0.425\linewidth]{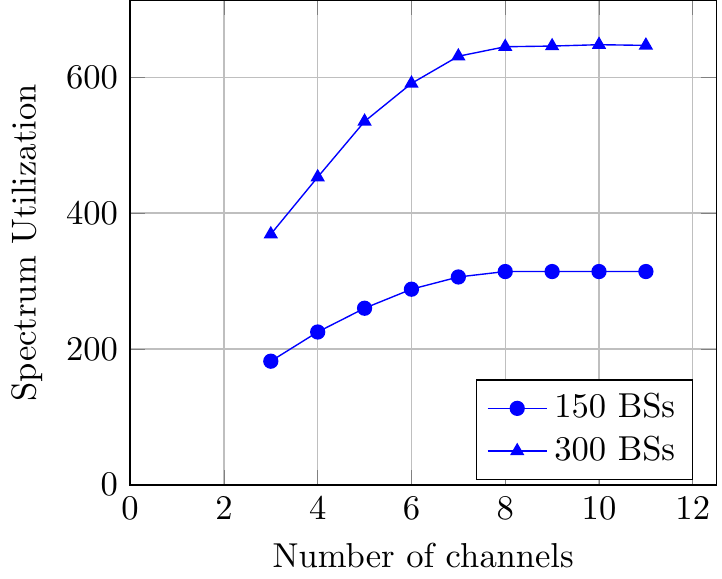}
 \caption{Comparison of spectrum utilization and number of channels for NUD-WSPAM in large networks\label{fig:c_vs_su_nu}.}
\end{figure}
\vspace{-0.3cm}

\section{Conclusions}
\label{sec:conclude}
In this paper, we have investigated the problem of spectrum allocation at operator level, for mutiple operator co-existence in a region. We consider 
multiple base stations to be associated with an operator to provide services to the end users. Therefore, an operator
has demand and valuation corresponding to each BS associated with it. 
To address the issue of multiple valuations at an operator, we have modeled the spectrum allocation problem among non-cooperative operators in 
multi-parameter environment with an objective of maximizing the social welfare of the system.
First, we propose a strategy-proof mechanism for single channel demand across BSs of co-existing operators.
Then we extend it for multiple channels considering non-uniform
demand across the BSs of the operators. 
We prove that the mechanisms SC-SPAM and NUD-WSPAM are guaranteed to be strategy-proof and weakly strategy-proof, respectively.
The performances of the proposed algorithms are evaluated using Monte Carlo simulations and compared with those of the other existing
mechanisms. The
performances of the proposed mechanisms are near optimal in terms of spectrum utilization and social welfare.
 Furthermore,  the analysis of computational complexity reveals that the proposed mechanisms are implementable in large networks in real
time scenarios. Thus, the proposed mechanisms solve the issue of intractability arising in VCG mechanism.

\section*{Appendix}
\subsection{Proof of Theorem~\ref{thm:thm1}}\label{app:a}

To show the strategy-proofness of the algorithm, possible scenarios can be divided into two categories:\\
\textit{\underline{Scenario 1}} : A operator $i$ tries to deviate from truthfulness by bidding greater than the true valuation, i.e., 
$\sigma_i^{b} > \sigma_i^{v}$.\\
 \textit{Case (i)}: Operator $i$ does not win the channel even after bidding untruthfully at $\sigma_i^{b}$, greater than  $\sigma_i^{v}$.
 Hence, it will have utility, $\mathcal{U}_i = 0$.\\
 \textit{Case (ii)}: Operator $i$ wins the channel at its bidding valuation $\sigma_i^{b}$ (which is greater than the true valuation) 
 as well as its true valuation $\sigma_i^{v}$. It will have positive utility, $\mathcal{U}_i = \sigma_i^{v} - p_i$, which is same as in the case 
 operator bids at the true valuation. Thus, bidding at higher valuation does not lead to any extra incentive.\\
 \textit{Case (iii)}: Operator $i$ wins channel at $\sigma_i^{b}$, but looses at $\sigma_i^{v}$. Here, Operator $i$ gets channel on higher
  bid (by misreporting) which is greater than its critical operator bid (Algorithm \ref{channel_allocation}). But, it has to pay higher price 
  which results in negative utility.
  
 \begin{equation*}
 \begin{split}
   \mathcal{U}_i &= \sigma_i^{v} - p_i,\\
   &= \sigma_i^{v} - \sigma_i^{c} \quad \text{where } p_i = \sigma_i^{c},\\
   &\leq 0. \quad (\because \sigma_i^{v} < \sigma_i^{c}).
  \end{split}
 \end{equation*}
  

\noindent\textit{\underline{Scenario 2}} : Operator $i$ tries to deviate from truthfulness by bidding less than the true valuation, i.e., 
$\sigma_i^{b} < \sigma_i^{v}$.\\
 \textit{Case (i)}: Operator $i$ looses the channel at $\sigma_i^{b}$ as well as its true valuation, $\sigma_i^{v}$. Thus, it will have 
 $\mathcal{U}_i = 0$.\\
  \textit{Case (ii)}: Operator $i$ wins the channel at $\sigma_i^{b}$ as well as its true valuation, $\sigma_i^{v}$ which follows from 
 monotonicity. Thus, it will have $\mathcal{U}_i = \sigma_i^{v} - p_i$.\\
 \textit{Case (iii)} : Operator $i$ looses at $\sigma_i^{b}$, but wins bidding at $\sigma_i^{v}$. Thus, the operator suffers loss by deviating
  to untruthful value with zero utility. However, bidding at $\sigma_i^{v}$ results in  channel allocation  to operator $i$ with non-negative 
  utility $\mathcal{U}_i = \sigma_i^{v} - p_i$.

From the above scenarios, it can be seen that bidding at $\sigma_i^{b} \neq \sigma_i^{v}$, does not improve the utility of operator. 
Thus, $\sigma_i^{b} = \sigma_i^{v}$ is the \textit{weakly dominant strategy} for operator $i$. This completes the proof.

\subsection{Proof of Theorem~\ref{thm:thm3}}\label{app:b}

 To prove the strategy-proofness, we are required to show that the deviation from the true valuation for any operator can never increase the 
 utility.
 We consider two scenarios: $(1)$ if an operator bids at a value higher than the true value, and $(2)$ if an operator bids at a value less than
 the true valuation.\\
Let $\sigma_i^{\rm t}$ denote the sum of the true valuations at the BSs of operator $i$ for the allocated channels.\\
Let $\beta_i^{\rm t}$ denote the sum of the bids of the channels allocated across the BSs of operator $i$.\\
Critical valuation, utility and final conflict graph at $\beta_i^{\rm t}$ are denoted as $\tilde{\sigma}_i^c$, $\tilde{U}_i$ and 
$\mathcal{\tilde{G}}^f$, respectively.
Let $x_i^f$ and $\tilde{x_i}^f$ denote the final allocation vector of operator $i$ with bids $\sigma_i^{\rm t}$ and $\beta_i^{\rm t}$, respectively.
Further, we define $\tilde{x_i}^f > x_i^f$, if $\exists$ at least a BS $\ell$ such that $\tilde{x_i}^f(\ell)> x_i^f(\ell)$.\\
 \noindent \textit{\underline{Scenario $1$}} : The operator bid is more than the true valuation of the channels allocated at its BSs,
 $\sigma_i^{\rm t} < \beta_i^{\rm t}$.
 Here, again we may have following cases:\\
 \textit{Case $(i)$}: Final allocation vector for all operators remains unchanged i.e., $\tilde{x_i}^f = x_i^f$,
 $\forall$ $i$.  
 Therefore, $C(i)$ and $\sigma_i^c$ for operator $i$ remains the same even at $\beta_i^{\rm t}$.
 Hence, operator utility $\mathcal{U}_i$ remains the same.\\
 \textit{Case $(ii)$}: Operator $i$ is allocated more number channels i.e., $\tilde{x_i}^f > x_i^f$. Since supply is limited, number of 
 channels allocated to some operators other than $i$ decreases i.e.,
 $\tilde{x_j}^f < x_j^f$ such that $j \neq i$.
 Let us say, operator $i$ is allocated extra channels in iteration $k$. Then, $\sigma_i^b(k) > \sigma_j^b(k)  > \sigma_i^v(k) $ for $j \neq i$. 
 However, at true value unsatisfied BSs of operator $i$ are not allocated channel and are present in $\mathcal{G}^f$. Due to untruthful bidding 
 of operator $i$, $\mathcal{\tilde{G}}^f$ comprise of the BSs of operator $j$ with higher aggregate true valuation. Therefore,  $\tilde{\sigma}_i^c > \sigma_i^c$, this implies 
 $\tilde{\mathcal{U}}_i < \mathcal{U}_i$. Hence, deviation from true value does not increases utility of operator $i$. 
\\
 \textit{Case $(iii)$}: Operator $i$ is allocated less number channels i.e., $\tilde{x_i}^f < x_i^f$.
 This is not possible due to monotonicity (Lemma \ref{lemma:monotone2}).
 
 \noindent \textit{\underline{Scenario $2$}} : The operator bid is less than the true valuation of the channels allocated at its BSs, i.e.,
 $\sigma_i^{\rm t} > \beta_i^{\rm t}$.\\
  \textit{Case $(i)$}: The number of channels allocated across the BSs and the final allocation vector remains unchanged.\\
  With the similar argument as in Case $(a)$ of Scenario $1$. The utility of the operator $i$ does not change.\\
  \textit{Case $(ii)$}: Operator $i$ is allocated more number channels i.e., $\tilde{x_i}^f > x_i^f$. 
  This is not possible due to monotonicity (see Lemma \ref{lemma:monotone2}).\\
  \textit{Case $(iii)$}: Operator $i$ is allocated less number channels i.e., $\tilde{x_i}^f < x_i^f$.
  As the number of channels allocated decrease on deviation from the $\sigma_{i}^{t}$, operator $i$ suffers loss.
  
 Thus, we establish that the deviation from true valuation does not lead to utility gain. Therefore, the proposed algorithm 
  is weakly strategy-proof.

\bibliographystyle{IEEEtran}
\bibliography{final_refer}

\end{document}